\pdfoutput=1
\documentclass[12pt,a4paper]{article}
\usepackage{mathtext, amsmath, amsthm, amssymb,latexsym}
\usepackage[T1,T2A]{fontenc}
\usepackage[utf8]{inputenc}
\usepackage[english]{}
\usepackage{url}
\usepackage{color}

\theoremstyle{definition}
\newtheorem*{definition}{Definition}
\newtheorem{theorem}{Theorem}
\newtheorem{prep}{Proposition}
\newtheorem{lemma}{Lemma}
\theoremstyle{remark}
\newtheorem*{remark}{Remark}
\usepackage{enumerate}
\usepackage{pdfpages}

\begin{document}
\title{On Circuit Complexity of Parity and Majority Functions in Antichain Basis}
\date{}
\author{Olga Podolskaya\thanks{Department of Discrete Mathematics, Faculty of Mechanics and Mathematics, Moscow State University. E-mail: olgavikonov@gmail.com.}}
\maketitle 
\begin{abstract}
We study the circuit complexity of boolean functions in a certain infinite basis. The basis consists of all functions that take value $1$ on antichains over the boolean cube. 
We prove that the circuit complexity of the parity function and the majority function of $n$ variables in this basis is $\lfloor \frac{n+1}{2} \rfloor$ and $\left\lfloor \frac{n}{2} \right \rfloor +1$ respectively. We show that the asymptotic of the maximum complexity of $n$-variable boolean functions in this basis equals $n.$
\end{abstract}

\subsection{Introduction}\label{vvedeine}
A \textit{parity function} $p_n(x_1,\ldots, x_n)$ is a boolean function such that it is equal to $1$ iff \mbox{$x_1+\ldots +x_n\equiv 1\pmod{2}.$} A \textit{majority function} $m_n(x_1,\ldots,x_n)$ is a boolean function  such that it is equal to $1$ iff \mbox{$x_1+\ldots+x_n\geq \left \lceil \frac{n}{2}\right \rceil.$}

Consider a boolean cube with the coordinate-wise partial order on its vertices. An \textit{antichain} is such a subset of the boolean cube that no two of its
tuples are comparable. An \textit{antichain function} is a characteristic function of an antichain over the boolean cube. We denote by $AC$ the infinite complete basis consisting of all antichain functions of any number of variables~\cite{k-z1}. 

A circuit computing a boolean function in $AC$ basis we call \textit{$AC$ circuit}. In this paper we consider $AC$ circuits computing the parity function and the majority function. The definition of a circuit and other notation see in~\cite{lupanov, wegener}.

For a circuit $S$, we denote by $L(S)$ the \textit{complexity of this circuit} which is the number of its gates. For a function $f$, we denote by $L(f)$  the \textit{complexity of this function} which is the minimum number of gates required by a circuit to compute $f.$ The \textit{Shannon function} $L(n)$ is the minimum number of gates sufficient to compute any function of $n$ variables. 

The circuit complexity in $AC$ basis has been studied in papers~\cite{k-z1,k-z2,pod1,pod2}. The upper bound for the Shannon function $L(n)\leq n+1$ for all $n$ was given in~\cite{k-z1}. This bound was improved in~\cite{pod2}, where was shown the upper bound $L(n) \leq n$ for all $n.$ 

Concerning lower bounds, they have been studied in~\cite{k-z1, k-z2}. In~\cite{k-z1} O.\,M.\,Kasim-Zade proved the bound $L(p_n)\geq \Omega(n^{1/3}).$ Later, in~\cite{k-z2}, by extending of the method of paper~\cite{k-z1}, he obtained a stronger result: $L(p_n)\geq \Omega((n/\ln n)^{1/2}).$ This result was improved in~\cite{pod1}, where the lower bounds $\Omega(\sqrt n)$ for the complexity of the parity function, the majority function and almost all boolean functions of $n$ variables were shown.

The aim of this paper is to prove the following results. For all natural $n$ we prove that $L(p_n)=\lfloor \frac{n+1}{2} \rfloor,$ $L(m_n)=\left\lfloor \frac{n}{2}\right \rfloor +1.$  Consequently, we obtain the asymptotic of the Shannon function: $L(n)=\Theta(n).$

The paper is organized in the following way. In Sections~\ref{lower proof} and~\ref{lower majority} we prove that the complexity of an AC circuit computing the parity function $p_n(x_1,\ldots, x_n)$ is at least $\left \lfloor \frac{n+1}{2} \right\rfloor$ and the majority function --- at least $\left\lfloor \frac{n}{2} \right \rfloor +1$ for all $n.$ In Section~\ref{the idea} the idea of the proof is shortly described. In Section~\ref{upper proof} using the ideas from paper~\cite{pod2} we prove that for all $n$ the complexity of the parity and majority functions of $n$ variables in $AC$ basis are at most  $\left \lfloor\frac{n}{2} \right \rfloor +1$ and $\left\lfloor \frac{n+1}{2}  \right \rfloor$ respectively.

\subsection{Preliminaries}\label{denomination}
In this section we introduce some notation.

Denote by $[n]$ the set of natural numbers $\{1, 2, \ldots, n\}.$ 

Сonsider $n$-dimensional boolean cube with the coordinate-wise partial order on its vertices. 
Denote by $\boldsymbol{x}$ any tuple $(x_1,\ldots,x_n)$ of the boolean cube.

Denote by $\boldsymbol{x}^P,$ where $P \subseteq [n],$ the tuple $\boldsymbol{x}$ such that $\forall k \in [n]$ $x_k=1 \Leftrightarrow k \in P.$

For a given circuit with $n$ nodes of zero fan-in, assign to these nodes \textit{input variables} $x_1, x_2,\ldots,x_n.$ Call these nodes \textit{inputs of the circuit} (or circuit inputs).

A circuit is called \textit{reduced} if no output of a gate is inputed twice in any input of some other gate. Note that any $AC$ circuit of complexity $s$ is equivalent to a reduced $AC$ circuit of complexity at most $s$ (see~\cite{k-z2}).

A numeration on circuit gates is called \textit{regular} if any input of a gate is connected to an output of a gate with smaller index or an input of the circuit. We can introduce a regular numeration on any circuit~\cite{lupanov2}. For any circuit, we fix an arbitrary regular numeration on its gates. From now on we study reduced circuits with the fixed regular numerations.

Consider an $AC$ circuit $S$ with the regular numeration on its $s$ gates. Denote by $e_k$ the gate with index $k$. Denote by $g_k$ the antichain function corresponding to the gate $e_k$. Denote by $h_k$ the function computed by a subcircuit with the output $e_k.$ On a given input $\boldsymbol x$ denote by $h_k(\boldsymbol x)$ the value of the $k$-th function $h_k$ that is, the value of the function $g_k$ applied to the values $g_i(\boldsymbol x),$ where $i<k.$ 

Call the maximal boolean $n$-tuple $\boldsymbol 1=(1,1, \ldots, 1)$ the \textit{top} tuple of the cube, and the minimal boolean $n$-tuple $\boldsymbol 0=(0,\,0,\,\ldots,\,0),$ the \textit{bottom} one.

\subsection{Lower Bound: Proof Idea}\label{the idea}
Consider arbitrary $AC$ circuit $S$ with $n$ inputs computing the parity function $p_n(x_1,\ldots, x_n)$ of $n$ variables. In this section we roughly describe the idea of the proof of the inequality $L(S)\geq \left \lfloor \frac{n+1}{2} \right\rfloor$.

Our proof technique is centered on a new notion of \textit{first non-zero gate} on a given input. Here we briefly explain the idea of the proof without technical detail, so right now we do not need an accurate definition of this notion. Thus a preliminary description is given below. 

Consider a circuit $S$ with gates $e_1,\ldots,e_s.$ For arbitraty tuple~$\boldsymbol{\alpha},$ let $\{e_{i_1},\ldots,e_{i_k}\}$ be the set of all gates such that for all $1\leq j \leq k$ 
\begin{enumerate}[{1.}]
\item $h_{i_j}(\boldsymbol{\alpha})=1$;
\item At least one input of the gate $e_{i_j}$ is connected directly to an input of the circuit.
\end{enumerate}

The gate $e_m$ with the minimal index among the gates $e_{i_1},\ldots,e_{i_k}$ initially we will call the first non-zero gate on the tuple $\boldsymbol{\alpha}.$ 

Informally, the difference between this preliminary description and the accurate definition of the notion is the following. In the description a so-called first non-zero gate is chosen among all the gates that output $1$ on a given input and for all these gates, at least one input is connected to an input of $S$. But in the accurate definition a first non-zero gate is chosen among the gates that output $1$ on a given input and also satisfy a certain property. The accurate definition is given below (see Definition in Section~\ref{accurate def}).

Now let us shortly describe the idea of the proof. We will construct a certain chain $C$ consisting of $n+1$ \,$n$-component tuples. We initially include the tuples $\boldsymbol{0}$ and $\boldsymbol{1}$ in the chain. Then we add tuple by tuple to the chain $C$ by the following algorithm. Assume a tuple $\boldsymbol {\beta}$ is added to the chain $C.$ Input $\boldsymbol {\beta}$ to the circuit. Let $e_k$ be the first non-zero gate on this tuple that is, it outputs $1$ on this input. By the second property of the above-mentioned description there is a component of the tuple $\boldsymbol {\beta}$ which is essential for the function $g_k.$ Let us flip the value of this component. Thus we get a new tuple $\boldsymbol {\gamma}$ which we also add to the chain. Since the gate $e_k$ corresponds to the antichain function, this gate can not be first non-zero one on $\boldsymbol {\gamma}.$ 

We construct the chain in such a way that no gates of the circuit $S$ can be first non-zero twice on the tuples of $C.$ Recall that the circuit $S$ computes the parity function. Thus if we input the tuples of $C$ in the circuit then first non-zero gates occur roughly speaking on half of the tuples of $C.$ This implies that the total number of gates in the circuit $S$ is equal to at least half of the size of $C.$ The constructing process is described in detail in Section~\ref{lower proof}.

\subsection{Lower Bound for Parity Function}\label{lower proof}
In this section we prove the following
\begin{prep}~\label{parity}
For all natural $n,$ $L(p_n)\geq \left\lfloor \frac{n+1}{2} \right\rfloor.$
\end{prep}

This section is organized as follows. In Subsection~\ref{construction} we construct the chain which plays a pivotal role in the proof of Proposition~\ref{parity}. Subsection~\ref{lemms' proofs} contains the proofs of the lemmas that we use during the analysis in Subsection~\ref{construction}.
\subsubsection{Construction}\label{construction}
In this subsection we construct a chain $C$ of size $n+1$ consisting of $n$-component tuples. Initially, we add the top and the bottom tuples  $\boldsymbol 1$ and  $\boldsymbol 0$ to the chain. Further, we describe how to obtain the rest $n-1$ tuples.

Informally, we construct the chain from both sides simultaneously. Namely, we descend from the top tuple and ascend from the bottom one by a certain rule. For clarity, we consider descending from the tuple $\boldsymbol 1$. By the algorithm mentioned in Section~\ref{the idea} we find a certain component of the tuple $\boldsymbol 1$ and flip its value, e.g. change the value of this component to $0.$ We obtain a new tuple. Clearly, this tuple is comparable to the tuple $\boldsymbol 1.$ So, we add this new one to the chain. From now on the component whose value was flipped is fixed to $0.$ So, we obtained a subcube whose maximal tuple is the new tuple that we get and the minimal one is the tuple $\boldsymbol 0.$ Then we consider the problem on the subcube. For any subcube, let us also call its maximal and minimal tuples the top and the bottom ones respecitvely. We ascend from the bottom tuple in the same way .

To construct the chain, we consider the gates of the circuit according to the regular numeration, starting from the gate $e_1.$ The constructing process has the following parameters.

\begin{itemize}
\item[--] Index $i$ is a number of the latest gate considered. Note that $i\leq s.$ Specifically, on the $i$-th step of the process we consider the gate $e_i$ and if necessary gates with smaller indices.
\item[--] Simultaneously with the chain we construct two sets: $F_i,\,T_i~\subseteq~[n].$ The set $F_i$ consists of the indices of the input variables that are fixed to $0$; $T_i$ respectively consists of those that are fixed to $1.$ 
We call by \textit{free variables} the input variables with indices from $[n]\setminus (F_i\cup T_i)$ (since their values are not yet fixed) and by \textit{free inputs} of the circuit (or free circuit inputs) the inputs that correspond to these free variables. On the $i$-th step we consider the subcube whose top tuple is $\boldsymbol{x}^{[n]\setminus {F_i}}$ and the bottom one is $\boldsymbol{x}^{T_i}.$ 
\item[--] The set $\mathfrak E_i$ consists of all gates $e_k,$ where $k<i,$ such that their inputs are not connected to the free inputs of the circuit. In other words, $\mathfrak E_i$ is a set of gates whose inputs may be connected to outputs of gates with indices smaller that $i$ and to the circuit inputs with indices from the set $F_i \cup T_i$. 

Note that for all $i:$ $\mathfrak E_i\subseteq \{e_1,e_2,\ldots,e_i\}$. We add a gate to $\mathfrak E_i$ when after the $i$-th step all of its inputs connected directly to the inputs of the circuit get fixed to $1$ or $0.$ Note also that with growing $i$ the sets  $F_i$, $T_i$ and $\mathfrak E_i$ are not decreasing.
\end{itemize}

Now we are ready to define a key notion that will be used in the proof.

\begin{definition}\label{accurate def}
Consider a circuit $S$ with gates $e_1,\ldots,e_s.$ For arbitraty tuple $\boldsymbol{\alpha}$ let $\{e_{i_1},\ldots,e_{i_k}\}$ be the set of all gates such that for all $1\leq j \leq k$
\begin{enumerate}[{1.}]
\item $h_{i_j}(\boldsymbol{\alpha})=1$;
\item $e_{i_j}\subseteq \{e_1,\ldots, e_i\} \setminus \mathfrak E_i$, that is,
at least one input of the gate $e_{i_j}$ is connected directly to a free input of the circuit.
\end{enumerate}

The gate $e_m$ with the minimal index among the gates $e_{i_1},\ldots,e_{i_k}$ is called the \textit{first non-zero gate} of the circuit on the tuple $\boldsymbol{\alpha}.$

\end{definition}
Let us stress again the difference between the preliminary description given in Section~\ref{the idea} and this accurate definition. In the definition from all the gates that output $1$ on $\boldsymbol \alpha$ we consider only such gates that some of their inputs are connected to some \textit{free} inputs of the circuit.

At the beginning of the constructing process we let $i=0$, $F_0=T_0=\emptyset$, $\mathfrak E_0=\emptyset$, $C=\{\boldsymbol 0, \boldsymbol 1\}.$

During the process the two following \textbf{properties} should be maintained.
\begin{enumerate}[{1.}]
\item $F_i\cap T_i =\emptyset.$ That is a variable can not be simultaneously fixed to $1$ and $0.$ 
\item\label{imp} For any gate $e_j \in \{e_1,e_2,\ldots,e_i\}\setminus \mathfrak E_i,$ for the function $h_j$ corresponding to this gate, we have: $h_j(\boldsymbol{x}^{[n]\setminus F_i})=0=h_j(\boldsymbol{x}^{T_i}).$
\end{enumerate}

Property~\ref{imp} means that after the $i$-th step if we input the top and the bottom tuples of the current subcube to the circuit then all the gates $e_k,$ where $k\leq i,$ $e_k\notin \mathfrak{E}_i,$ output $0$ on these inputs.

The constructing process has two stages. We start with Stage 1, then if the chain is not yet constructed we move to the Stage 2. In detail we explain this issue later. 

\textbf{Stage 1}. Now we describe the $(i+1)$-th step. Suppose first $i$ gates $e_1, e_2,\ldots, e_i$ are considered. We have constructed a certain set $F_i$ which includes the indices of the input variables that are fixed to $0.$ Similarly, there is a set $T_i$ consisting of the indices of the input variables that are fixed to $1.$ The current chain $C$ contains the tuples $\boldsymbol 0$ and $\boldsymbol 1$ as well as the tuples that we get simultaneously with constructing $F_i$ and $T_i$. Consider the current $(n - |F_i|- |T_i|)$-dimensional subcube. It has the top tuple $\boldsymbol{x}^{[n]\setminus F_i}$ and the bottom one $\boldsymbol{x}^T_i.$ Now consider the gate $e_{i+1}.$ All possible cases are considered below. Note they can not occur simultaneously.

\begin{enumerate}[{1.}]
\item \label{case 1}
Assume that inputs of the gate $e_{i+1}$ is connected directly to none of the circuit input with an index from the set $[n] \setminus (F_i \cup {T_i})$. We add the gate $e_{i+1}$ to the set $\mathfrak E_i.$ We obtain the new set $\mathfrak E_{i+1}.$ The $(i+1)$-th step is finished.
\item \label{case 2}
Otherwise, assume that at least one input of the gate $e_{i+1}$ is connected to a free circuit input. Thus the antichain function $g_{i+1}$ corresponding to the gate $e_{i+1}$ has at least one free essential variable $x_k$, where $k\in [n]\setminus (F_i\cup {T_{i}}).$ Let us verify whether Property~\ref{imp} is valid, i.e. whether the equation $h_{i+1}(\boldsymbol{x}^{[n] \setminus {F_i}})=h_{i+1}(\boldsymbol{x}^{T_i})=0$ is true.
    \begin{enumerate}[a.]
    \item \label{case a}
    If Property~\ref{imp} for the gate $e_{i+1}$ is valid then the $(i+1)$-th step is finished.
    \item \label{case b}
    Assume Property~\ref{imp} is not valid for the top tuple of the current subcube. That is, $h_{i+1}(\boldsymbol{x}^{[n]\setminus {F_i}})=1.$ In this case we fix to $0$ the value of the variable $x_k.$ Further, we let $F_{i+1}=F_i\cup \{k\},$ $T_{i+1}=T_i.$ Then we consider a subcube of dimension $(n-|F_{i+1}|-|{T_{i+1}}|).$ We add the top tuple $\boldsymbol{x}^{[n]\setminus F_{i+1}}$ of this subcube to the chain $C.$ 

Further, we verify whether Property~\ref{imp} is valid for all $e_j \in \{e_1,e_2,\ldots,e_i\}\setminus \mathfrak E_i$ on the $\boldsymbol{x}^{[n]\setminus F_{i+1}}.$ Let there be some gates that does not satisfy Property~\ref{imp} on this input. Let $e_l,$ where $l<i,$ be the first non-zero gate on $\boldsymbol{x}^{[n]\setminus F_{i+1}}$. 
\begin{lemma}\label{lemm 1} 
The antichain function $g_l$ which corresponds to the gate $e_l$ has at least one free essential variable. (See the proof in Section~\ref{lemms' proofs}.)
\end{lemma}
         
Consider a free essential variable of the function $g_l.$ Add the index of this variable to the set $F_{i+1}$. Thus we obtain a new set. Since we are still at the $(i+1)$-th step, denote this set again by $F_{i+1}.$ That is, we now switch to the subcube of smaller dimension. The new top tuple is yet another element of the chain $C$.

We fix to $0$ the input variables as described above in Case~\ref{case b} until one of the two following cases occurs: Property~\ref{imp} gets valid for all gates with indices $\leq i+1,$ or $F_{i+1}\cup T_{i+1}=[n]$ which means the chain is constructed.

Suppose we have fixed to $0$ yet another variable, and thus, the inputs of a gate, e.g. $e_k,$ $k\in [i+1],$ do no longer connected to any free inputs of $S.$ Then we add this gate to $\mathfrak E_{i},$ and let $\mathfrak E_{i+1}=\mathfrak E_{i}\cup\{e_k\}.$ 
 \begin{lemma}\label{lemm 2} 
Consider all tuples of the current subcube such that gates $e_l,$ where $l<k$, $e_l\in \{e_1,\ldots,e_{i+1}\} \setminus \mathfrak E_{i+1},$ output $0$ on them.  The above-described gate $e_k$ outputs $0$ on all these tuples. (See the proof in Section~\ref{lemms' proofs}.)
 \end{lemma} 
When everything described in Case~\ref{case b} is done the $(i+1)$-th step is finished. Note that the number of the step equals the index of the gate with which we start this step. If any gate with index smaller than $i+1$ occurs to be the first non-zero one during the step we deal with it inside the same step. Thus when the step is finished all the gates with indices $\leq i+1$ satisfy Property~\ref{imp} on the top and bottom tuples of the current subcube.
    \item \label{case c}
    Let Property~\ref{imp} be not valid on the bottom tuple of the subcube, i.e. $h_{i+1}(\boldsymbol{x}^{T_i})=1.$
    
\begin{lemma}\label{lemm 3} 
Case~\ref{case b} and Case~\ref{case c} can not occur simultaneously. (See the proof in Section~\ref{lemms' proofs}.)
\end{lemma}
		
		Case~\ref{case c} is symmetrical to~\ref{case b}, so, we do the same. Namely, we add the index $k$ to the set $T_i$, i.e. fix it to $1.$ Let $T_{i+1}=T_i\cup\{k\},$ $F_{i+1}=F_i.$ Thus we switch to a subcube of dimension $(n-|{F_{i+1}}|-|{T_{i+1}}|)$. We add the tuple $\boldsymbol{x}^{T_{i+1}}$ to the chain $C.$ Then, similarly as we do in the Сase~\ref{case b}, we add  to $T_{i+1}$ indices of variables that are essential for some other antichain functions with smaller indices. We repeat these procedure until one of the following cases take place: Property~\ref{imp} is valid for all gates $e_l,$ where $l\leq i+1,$ or all input variables are fixed to $1$ or $0$. 
    \end{enumerate}
\end{enumerate}

Note that during the whole process some of the gates may be added to $\mathfrak E_i.$

Now we proceed to Stage 2. Let us suppose that we performed the procedure described in Stage 1, and the $(n+1)$-tuple chain is constructed. Then the process is finished. Otherwise, we did $s$ steps described in Cases~\ref{case 1} and~\ref{case 2} of Stage 1. So, we considered all $s$ gates of the circuit, but we still do not have the chain of size $n+1$. We have constructed the sets $F_s$, $T_s$ and $\mathfrak E_s.$ For simplicity, from now on we will write $F,$ $T$ and $\mathfrak E$ respectively. 

\textbf{Stage 2.}
Note that the original $n$-dimensional cube is restricted to a subcube of dimension $(n-|F|-|T|) $ such that all the gates $e_1,\ldots,e_s,$ except those from the set $\mathfrak E,$ output $0$ on the top and bottom tuples of this subcube.

Then we add an arbitrary index from $[n]\setminus (F\cup T)$ to the set $T.$ That is, we fix to $1$ a certain variable. Thus we get a set $T'.$ Add the bottom tuple $\boldsymbol{x}^{T'}$ of the current subcube to the chain $C.$ Since the circuit computes the parity function, it is easy to see that $h_s(\boldsymbol{x}^T)\neq h_s(\boldsymbol{x}^{T'}).$ Therefore, there are some gates (possibly it is only $e_s$) that output different values on the inputs $\boldsymbol{x}^{T}$ and $\boldsymbol{x}^{T'}.$ Let $e_j$ be the gate with the minimal index among all the gates described. Clearly, $e_j\notin\mathfrak E,$ since otherwise it would output the same values on the above-mentioned inputs. For the gate $e_j$ we do steps from Stage 1, namely, Case~\ref{case c}. In detail, we flip the value of a free input variable assigned to the circuit input which is connected to an input of the gate $e_j.$ And fix this new value. We add a new obtained tuple to the chain $C.$ As long as $F\cup T\neq [n]$ we repeat these steps. When $F\cup T= [n]$ the chain is constructed.
\begin{remark}
Note that on the $i$-th step we input the tuples $\boldsymbol{x}^{T_i}$ and $\boldsymbol{x}^{[n]\setminus {F_i}}$ to the circuit. Since  $T_i\subseteq [n]\setminus {F_i},$ the tuples $\boldsymbol{x}^{T_i}$ and $\boldsymbol{x}^{[n]\setminus {F_i}}$ are comparable. Hence, the  constructed set of tuples is a chain, indeed.
\end{remark}
Constructing the chain $C$ we do the steps of Stage 1 until all the gates are considered. Then we switch to Stage 2. Now we stress the following. During the steps of Stage 1 we add a tuple to the chain iff we find a first non-zero gate. And we add such a tuple that this gate is not a first non-zero one on this new tuple. At Stage 2 we add every other tuple to the chain in order a gate to stop being first non-zero one. 

\begin{lemma}\label{lemm 4}
None of the gates $\{e_1,\ldots,e_s\}\setminus \mathfrak E$ can be first non-zero twice on the tuples of the chain $C.$
\end{lemma}
\begin{proof}[Proof of Lemma~\ref{lemm 4}]

Assume on the contrary, for two different $\boldsymbol{x}^P, \boldsymbol{x}^{P'} \in C$ there exists index $k$ such that $h_k(\boldsymbol{x}^P)=h_k(\boldsymbol{x}^{P'})=1,$ and for all $e_j \in \{e_1,\ldots, e_s\} \setminus \mathfrak E,$ where $j<k:$ $h_j(\boldsymbol{x}^P)=h_j(\boldsymbol{x}^{P'})=0.$
Consider the antichain function $g_k$ corresponding to the gate $e_k.$ By construction, without loss of generality, assume the tuple $\boldsymbol{x}^P$ is obtained from the tuple $\boldsymbol{x}^{P'}.$ Recall that we obtain it by flipping the value of at least one essential variable of the finction $g_k.$ This means that the function $g_k$ outputs $1$ on the two comparable tuples. This contradicts the fact that $g_k$ is an antichain function. 
\end{proof}

Therefore we have constructed the $(n+1)$-tuple chain $C$ such that for at least $\lfloor \frac{n+1}{2} \rfloor$ tuples there exist first non-zero gates and no gate is a first non-zero twice on the tuples of $C.$ This concludes the proof of inequality $L(p_n)\geq\lfloor \frac{n+1}{2} \rfloor.$ Proposition~\ref{parity} is proved.
\subsubsection{Proofs of Lemmas}\label{lemms' proofs}

\begin{proof}[Proof of Lemma~\ref{lemm 1}.]
We have $h_l(\boldsymbol{x}^{[n]\setminus F_{i+1}})=1,$ $h_l(\boldsymbol{x}^{T_{i+1}})=0.$
What $g_l$ outputs may depend on values on some circuit inputs and on the values that output the following gates:
\begin{enumerate}[{1)}]
\item $e_j,$ where $j<l$ and $e_l \in \{e_1,\ldots,e_{i+1}\} \setminus \mathfrak E_{i+1},$
\item $e_r,$ where $r<l$ and $e_r\in \mathfrak E_{i+1}.$
\end{enumerate}

The outputs of the gates of type 2 determine the outputs of the gates of type 1. Since $l$ is the minimal index, all gates $e_j$ of type 1 are equal to $0$ on the top and the bottom tuples of the current cube. Thus we obtain that $g_l$ outputs different values on the two tuples with the following property: in these tuples the components corresponding to the outputs of the gates with smaller indices are equal. Thus the other components corresponding to the input variables should be different. That is, there is at least one free variable among the input variables such that it is essential for the function $g_l.$
\end{proof}
\begin{proof}[Proof of Lemma~\ref{lemm 2}.] 
The function $g_k$ computed on the output of the gate $e_k$ has no free variables. Thus on all above-mentioned inputs the gate $e_k$ outputs the same value. This gate outputs $0$ on the bottom tuple of the cube. Therefore, it outputs $0$ on all the inputs described.
\end{proof}
\begin{proof}[Proof of Lemma~\ref{lemm 3}.]
Assume on the contrary that Cases~\ref{case b} and~\ref{case c} occur simultaneously. Functions $h_j$ such that $j<i+1$ and $e_j \notin \mathfrak E_i$ output $0$ on the tuples $\boldsymbol{x}^{[n]\setminus F_i}$ and $\boldsymbol{x}^{T_{i+1}}.$ These two tuples are comparable. By the construction, the function $g_{i+1}$ has an essential variable whose value is different in the tuples $\boldsymbol{x}^{[n]\setminus F_i}$ and $\boldsymbol{x}^{T_{i+1}}.$ Hence we have shown that the function $g_{i+1}$ outputs $1$ on two comparable tuples. This contradicts the fact that $g_{i+1}$ is an antichain function.
\end{proof}

\subsection{Lower Bound for Majority Function}\label{lower majority}

\begin{prep}\label{majority}
For all natural $n,$ $L(m_n)\geq \left\lfloor \frac{n}{2} \right \rfloor +1.$
\end{prep}
\begin{proof}
The beginning of the proof is the same as that for the parity function in Section~\ref{lower proof}. That is, for a given circuit computing the $n$-variable majority function we construct an $(n+1)$-tuple chain with the same properties. The first stage of the constructing process is exactly the same as in the proof of Proposition~\ref{parity}.
 
 Then suppose we have partially constructed a chain and obtained its part of size $l.$ If $l=n+1$ then Proposition~\ref{majority} is proved. Suppose $l< n+1.$ The starting $n$-dimensional cube is restricted to an $(n-l+2)$-dimensional subcube. Let $\boldsymbol{\alpha}=(\alpha_1,\ldots,\alpha_n)$ and $\boldsymbol{\beta}=(\beta_1,\ldots,\beta_n)$ be respectively the top and the bottom tuples of this subcube. Recall that the last gate $e_s$ outputs the value of the majority function on a given input. By the construction of the chain, each gate outputs the same value on both inputs $\boldsymbol\alpha$ and $\boldsymbol\beta.$ So clearly, $h_s(\boldsymbol{\alpha})=h_s(\boldsymbol{\beta}).$ 

Let $a=h_s(\boldsymbol{\alpha})=h_s(\boldsymbol{\beta}).$ It is easy to see that if we input a tuple from the chain, and the circuit outputs the value $1-a$ on that input then there is a first non-zero gate on this tuple. If $a=0$ there are $\left \lfloor \frac{n}{2}\right\rfloor +1$ such tuples; in this case Proposition~\ref{majority} is proved. If $a=1$ there are $\left\lceil \frac{n}{2}\right\rceil$ such tuples. Therefore, there are exactly the same number of first non-zero gates in the circuit corresponding to these tuples. None of these non-zero gates is the gate $e_s.$ Thus the total number of the circuit gates is at least $\left\lceil \frac{n}{2} \right\rceil +1\geq \left\lfloor \frac{n}{2}\right\rfloor +1.$
\end{proof}
\subsection{Upper Bounds for Majority and Parity Functions}\label{upper proof}
In this section we obtain the upper bounds for the parity and majority functions.
\begin{prep}\label{prep 3}
For all natural $n,$\,
$
    L(m_n) \leq \left\lfloor \,
    \frac{n}{2} \, \right \rfloor +1.
$
\end{prep}
\begin{proof}
Consider $n$-dimensional boolean cube. Recall that the majority function of $n$ variables is equal to $1$ only on tuples that have at least $\left \lceil \frac{n}{2} \right \rceil$ ones. Denote by $M$ the support of the function $m_n$.

Let us call by the \textit{t-th layer} of the boolean cube the set of $n$-tuples with $t$ ones. For all $t=0,1,\ldots, n$ the t-th layer is an antichain over the boolean cube. Denote these antichains by $A_t.$ Clearly, $M=\bigsqcup_i A_i,$ where $i={\left \lceil
\frac{n}{2} \right \rceil}, {\left \lceil  \frac{n}{2} \right \rceil + 1}, {\left \lceil
\frac{n}{2} \right \rceil + 2}, \ldots, n.$ 

Define the functions  $h_t(x),$ where $x=(x_1, \ldots, x_n),$ $t=0, 1,\ldots,n,$ as follows:
$$h_t(x)= 1 \Leftrightarrow \sum\limits_{k=1}^n {x_k}= {t},
$$

It is easy to see that each $h_t$ is a characteristic function of an antichain $A_t.$

Let us denote $(y_{{\left \lceil  \frac{n}{2} \right \rceil}}, \ldots, y_{n-1})$ by $\boldsymbol y$ and $(y_{{\left \lceil  \frac{n}{2} \right \rceil}}, \ldots, y_{n-1},x_1,\ldots,x_n)$ by $(\boldsymbol y, \boldsymbol x).$ 

Let function $g(\boldsymbol y, \boldsymbol x)$ be equal to $1$ iff $(\boldsymbol y,\boldsymbol x)\in M_1\sqcup M_2,$ where $M_1,$ $M_2$ are defined by the following
$$\mbox{${(\boldsymbol y,\boldsymbol x})\in M_1\iff$}
\begin{cases}\label{nabor1}
    \exists\, j, \text{ such that } y_j=1,\\
    \forall i\neq j\,\, y_i=0,\\
    \sum\limits_{k=1}^{n} x_k=j,\\
\end{cases}
\mbox{$(\boldsymbol y,\boldsymbol x)\in M_2 \iff$}
\begin{cases} \label{nabor2}
    \forall\, i\,\,y_i=0,\\
    \sum\limits_{j=1}^{n} x_j= n.\\
\end{cases}
$$

It can be easily checked that $g\in AC.$ Indeed, consider two different arbitraty tuples from the support of the function $g:$ $(\boldsymbol {y_1}, \boldsymbol {x_1}) \neq (\boldsymbol {y_2}, \boldsymbol {x_2}).$ If $(\boldsymbol {y_1}, \boldsymbol {x_1}), (\boldsymbol {y_2}, \boldsymbol {x_2})\in M_1$ and $\boldsymbol y_1\neq \boldsymbol y_2$ then tuples are not comparable. If $(\boldsymbol {y_1}, \boldsymbol {x_1}), (\boldsymbol {y_2}, \boldsymbol {x_2})\in M_1$ and $\boldsymbol y_1= \boldsymbol y_2$ then the positions of the components, where $j$ ones of the tuple $\boldsymbol x_1$ occur, do not match those of the tuple $\boldsymbol x_2$. Thus the tuples $(\boldsymbol {y_1}, \boldsymbol {x_1})$ and $(\boldsymbol {y_2}, \boldsymbol {x_2})$ are not comparable. If, without loss of generality, $(\boldsymbol {y_1}, \boldsymbol {x_1})\in M_1,$ $(\boldsymbol {y_2}, \boldsymbol {x_2})\in M_2$ then $\boldsymbol y_1> \boldsymbol y_2$ and $\boldsymbol x_2>\boldsymbol x_1.$ Therefore, the tuples $(\boldsymbol {y_1}, \boldsymbol {x_1})$ and $(\boldsymbol {y_2}, \boldsymbol {x_2})$ are not comparable. Since the set $M_2$ contains only one tuple, the case when $(\boldsymbol {y_1}, \boldsymbol {x_1}), (\boldsymbol {y_2}, \boldsymbol {x_2})\in M_2$ is impossible.

Further, let us compute the majority function $m_n(\boldsymbol x)$ as follows:
\begin{equation}\label{eq}
m_n(\boldsymbol x)=g\left(h_{\left \lceil  \frac{n}{2} \right \rceil}(\boldsymbol x), \ldots,h_{n-1}(\boldsymbol x),x_1,\ldots,x_n\right).
\end{equation}

We will check now the equation~\eqref{eq}. Consider arbitrary tuple $\boldsymbol \alpha =(\alpha_1,\ldots,\alpha_n)$. There exist two possibilities:
\begin{enumerate}[{1)}]
\item $m_n(\boldsymbol{\alpha})=1.$ Thus $\sum\limits_{j=1}^{n} x_n\geq \left \lceil \frac{n}{2} \right \rceil.$ If $\left \lceil \frac{n}{2} \right \rceil \leq \sum\limits_{j=1}^{n} x_j\leq n-1$ then there exists $t=j$ such that $h_j (\boldsymbol \alpha)=1$ and for all $k\neq j:$ $h_k (\boldsymbol \alpha)=0$. Hence, the tuple $(\boldsymbol y,\boldsymbol \alpha),$ where $y_i=h_i(\boldsymbol \alpha),$ is contained in the set $M_1.$ Thus $g(\boldsymbol y,\boldsymbol \alpha)=1.$ If $\sum\limits_{j=1}^{n} x_j=n$ then for all $i\in\{\left \lceil \frac{n}{2} \right \rceil,\ldots, n-1\}:$ $y_i=h_i(\boldsymbol \alpha)=0.$ Therefore, $(\boldsymbol y,\boldsymbol \alpha)\in M_2,$ and then $g(\boldsymbol y,\boldsymbol \alpha)=1.$

\item $m_n(\boldsymbol \alpha)=0.$ Then for all $t\in\{\left \lceil \frac{n}{2} \right \rceil,\ldots,n-1\}:$ $h_t(\boldsymbol \alpha)=0.$ Thus, for $y_t=h_t(\boldsymbol \alpha),$ the tuple $(\boldsymbol y,\boldsymbol \alpha)$ is not contained in the sets $M_1$ and $M_2.$ Therefore, $g(\boldsymbol y,\boldsymbol \alpha)=0.$

\end{enumerate}

Thus, by~\eqref{eq}, we have shown the upper bound $\left \lfloor \, \frac{n}{2} \, \right \rfloor +1$ for the $n$-variable majority function.
\end{proof}
\begin{prep}\label{prep 4}
For all natural $n,$\,
$
    L(p_n) \leq \left\lfloor\frac{n+1}{2}\right \rfloor.
$
\end{prep}
\begin{proof}
In the same way as in the proof of Proposition~\ref{prep 3} we can compute the $n$-variable parity function by an $AC$ circuit of complexity $\left\lfloor \,\frac{n+1}{2} \, \right \rfloor.$ To do that, we define a similar function $g$ as in the proof of Proposition~\ref{prep 3}. Into this function we input the variables as well as the functions $h_t,$ where indices $t$ are equal to the numbers of certain alternating layers of the boolean cube. Namely, these are numbers of such alternating layers whose disjoint union gives the support of the parity function.
\end{proof}
The results of this paper can be summarized is the following theorems. 
\begin{theorem}~\label{th1}
For all natural $n,$ 
$$L(p_n)=\left\lfloor \frac{n+1}{2} \right\rfloor,\, L(m_n)=\left\lfloor \frac{n}{2}\right \rfloor +1.$$
\end{theorem}
\begin{proof}
Propositions~\ref{parity},~\ref{prep 3} proved in Sections~\ref{lower proof},~\ref{upper proof} together give the first equality for the parity function. Propositions~\ref{majority},~\ref{prep 4} proved in Sections~\ref{lower majority},~\ref{upper proof}, in turn, show the second equality for the majority function.
\end{proof}
These results imply the following theorem.

\begin{theorem}~\label{th2}
For all natural $n,$ $ L(n)=\Theta(n).$ 
\end{theorem}
\begin{proof}
From Theoreom~\ref{th1} we clearly obtain the bound $L(n)\geq \frac{1}{2}n$ for all natural $n. $ Also in~\cite{pod2} is shown the upper bound for the Shannon function: $L(n)\leq n.$ These two bounds together prove the asymptotic of the Shannon function: $L(n)=\Theta(n)$ for all $n.$ 
\end{proof}

\par\textbf{Acknowledgements.} The author is grateful to O.\,M.\,Kasim-Zade for stating the problem and for permanent attention to the work.

{\small
\bibliographystyle{abbrv}
\bibliography{bib/bibip}
}


\end{document}